\documentclass[conference,letterpaper]{IEEEtran}

\usepackage[utf8]{inputenc} 
\usepackage[T1]{fontenc}
\usepackage{url}
\usepackage{ifthen}
\usepackage{cite}
\usepackage[cmex10]{amsmath} 
\usepackage{amssymb,amsfonts}
\usepackage{mathtools}
\usepackage{algorithmic}
\usepackage{graphicx}
\usepackage{textcomp}
\usepackage{xcolor}
\usepackage{bm}
\usepackage{algorithm}
\usepackage{amsthm}
\theoremstyle{definition}
\newtheorem{definition}{Definition}
\newtheorem{remark}{Remark}
\theoremstyle{plain}
\newtheorem{proposition}{Proposition}
\newtheorem{theorem}{Theorem}

\newtheorem{corollary}{Corollary}

\newcommand{\N}{\mathbb{N}}
\newcommand{\F}{\mathbb{F}}

\newcommand{\1}{\mathbf{1}}

\usepackage{booktabs}
\usepackage{array}
\usepackage{multirow}
\usepackage{multicol}

\usepackage[top=1.047in, bottom=1.04in, left=0.673in, right=0.673in]{geometry}

\begin{document}
\title{Task-Oriented Boolean Function Computation: Practical Code Constructions} 

\author{%
  \IEEEauthorblockN{Yangshuo~He and Guanding~Yu}
  \IEEEauthorblockA{College of Information Science \& Electronic Engineering \\
                    Zhejiang University\\
                    Hangzhou\\
                    Email: \{sugarhe, yuguanding\}@zju.edu.cn}
  \and
  \IEEEauthorblockN{Jingge~Zhu}
  \IEEEauthorblockA{Department of Electrical and Electronic Engineering\\ 
                    University of Melbourne\\
                    Melbourne\\
                    Email: jingge.zhu@unimelb.edu.au}
}

\maketitle

\begin{abstract}
   Task-oriented communication conveys information that is necessary for downstream tasks. For binary decision tasks, this paradigm is information-theoretically formalized by Boolean function computation (BFC) via channels, where the receiver aims to determine the value of a function unknown to the transmitter. In this paper, we devise a practical code construction for the BFC problem based on a Reed--Solomon code. For noiseless binary channels, we derive finite-blocklength worst-case error bounds. By defining a rate function that captures how supported message length scales with channel uses, we characterize the rate-reliability tradeoff for different Boolean function families. With respect to this rate function, the proposed construction achieves an asymptotic computation rate of $1/2$. We further extend this construction to noisy channels by packing multiple BFC tasks into a single block and concatenating them with a conventional channel code. The corresponding finite-blocklength guarantees are expressed in terms of effective channel uses per function evaluation. With a channel code rate $R_c$, this construction achieves an asymptotic computation rate of $R_c/2$, yielding $C/2$ when capacity-achieving channel codes are employed. Numerical results illustrate the derived bounds and demonstrate substantial performance gains over conventional transmission. As examples, the proposed scheme achieves SNR coding gains of approximately $3.4$ and $6.4$~dB at an average task error probability of $10^{-5}$ for the exact-weight and rank-test tasks, respectively.
\end{abstract}

\section{Introduction}
\label{sec:introduction}
In many edge communication systems, instead of reconstructing the entire transmitted message, the receiver only requires partial information to execute a specific task. This motivates a paradigm shift from message transmission to task-oriented communications, where performance is measured by task accuracy rather than bit error rate~\cite{beyond-bits,task-oriented}. A fundamental and practical class of tasks is binary decisions, such as membership and threshold tests. These can be viewed as evaluating a Boolean function that represents the propositional formula, connecting communication with logical decision tasks~\cite{propositional-logic}.

Boolean function computation (BFC) via channels formalizes this setting~\cite{BFC,BFC-capacity}. Specifically, the transmitter owns a binary message $i$, while the receiver holds a function $f$ from a known family and aims to determine the function value $f(i)$. Importantly, the transmitter has no information about the selected function $f$, thus directly transmitting the result is infeasible. The primary problem in BFC via channels is determining the number of channel uses required to reliably compute a query for a given message length.

Identification via channels~\cite{ID} is a special case of BFC, where the receiver determines whether the transmitter's message matches a particular message. With a randomized encoder, a number of messages doubly exponential in the channel blocklength can be reliably identified. 
This differs fundamentally from the exponential growth in conventional transmission. Recent work in~\cite{ID-task} also motivates identification through task-oriented communication in a joint identification and sensing scenario with noisy feedback.
Meanwhile, BFC covers a broader collection of binary decision tasks. Achievability and converse results for Boolean function families characterized by Hamming weights are established in~\cite{BFC-capacity}. Via the notion of \textit{rate function}, their results characterize the fundamental limit on how the message length can scale with the channel uses asymptotically under BFC. However, these asymptotic results fail to provide an explicit finite-blocklength coding scheme.

Reed--Solomon (RS) codes provide an explicit method for code construction in identification via channels. Early work~\cite{CWC} established capacity-achieving identification codes based on constant-weight constructions. Polynomial universal-hash constructions for identification appeared in~\cite{universal-hashing}, and the RS tagging representation was developed for watermark identification in~\cite{RS-tagging}. Later works analyzed capacity-achieving concatenated RS tagging codes~\cite{RS-performance} and implemented them on software-defined radios~\cite{RS-implementation}. For noisy channels, the noiseless RS tagging code is commonly concatenated with a standard channel code~\cite{RS-tagging,ID-tutorial}. However, the existing RS tagging constructions are tailored for identification and do not directly yield finite-blocklength codes for general BFC problems.

In this paper, we develop the RS tagging code construction for BFC via channels. Over a noiseless binary channel, we derive a finite-blocklength worst-case error bound and a rate-reliability tradeoff. We then prove that our code construction attains an asymptotic achievable computation rate with the same optimal rate function identified in~\cite{BFC-capacity}. For noisy channels, we extend this construction by packing multiple tagging codes into one block and concatenating with a standard channel code. Both finite-blocklength and asymptotic results are derived in terms of the number of effective channel uses per BFC message. Simulations demonstrate the effectiveness of the proposed code constructions, showing significant SNR coding gains of $3.4$~dB for the exact-weight task and $6.4$~dB for the rank-test task over conventional full-message transmission.

\section{BFC Problem Formulation}
\label{sec:system_model}
Let $i\in\mathcal{I}_m:=\{0,1\}^m$ denote the binary message of length $m$ at the transmitter. The receiver selects a Boolean function $f$ from a family $\mathcal F_m:=\{f:\mathcal{I}_m\to\{0,1\}\}$. Although the function family $\mathcal F_m$ is shared by both transmitter and receiver, the particular $f$ is unknown to the transmitter. Instead of estimating the complete message $i$, the receiver only needs to recover $f(i)$. Thus, $f$ represents the receiver's task.

Let $f^{-1}(1):=\{i\in\mathcal{I}_m:f(i)=1\}$ denote the pre-image of 1 under $f$. Then, the Hamming weight of $f$ is defined as
\begin{equation}
    S_f:=|f^{-1}(1)|.
\end{equation}
We primarily consider the class of functions with Hamming weight no larger than $S$, denoted by
\begin{equation}
    \begin{split}
    \mathcal F_m(\leq S) & :=\{f:\mathcal{I}_m\to \{0,1\}:S_f\leq S\}.
    \end{split}
\end{equation}
In particular, computing the functions in $\mathcal F_m(1)$ is equivalent to the classic identification via channels problem~\cite{ID}.

Consider a channel $W(\cdot|x^n)$ over space $\mathcal{Y}^n$ given input codeword $x^n\in\mathcal{X}^n$.
As defined in~\cite{BFC}, an $(n,m,\mathcal F_m,\lambda_1,\lambda_2)$ BFC code consists of a stochastic encoder $Q_i$ on codeword space $\mathcal{X}^n$ for all $i\in\mathcal{I}_m$ and a decoding region $D_f\subseteq \mathcal{Y}^n$ for all $f\in \mathcal F_m$, such that the false-negative (FN) and false-positive (FP) probabilities are bounded by

\begin{align}        
    P_{\rm FN}:=\max_{\substack{i,f:f(i)=1}}
    \sum_{x^n}Q_i(x^n)W(D_f^{\rm c}|x^n)
    &\leq \lambda_1, \label{equ:BFC_FN}\\
    P_{\rm FP}:=\max_{\substack{i,f:f(i)=0}}
    \sum_{x^n}Q_i(x^n)W(D_f|x^n)
    &\leq \lambda_2. \label{equ:BFC_FP}
\end{align}

Unlike conventional transmission, the maximal message length that can be reliably supported need not scale linearly with the channel uses $n$ in the BFC problem~\cite{BFC,BFC-capacity}. In fact, the scaling behavior of $m$ depends on the Boolean function family. Therefore, we utilize a general rate function to characterize the asymptotic performance.
\begin{definition}
\label{def:finite_rate}    
Let $m,n\in \N$ and $L:\mathbb{R}^{+}\times\mathbb{N}\to\mathbb{R}^{+}$
be a rate function. The computation rate of an $(n,m,\mathcal F_m,\lambda_1,\lambda_2)$ BFC code is $R$ with the rate function $L$ if $m =L(R,n)$.
\end{definition}
The function $L$ specifies the scaling behavior of the supported message length, while $R$ is the leading constant.
For example, $L(R,n)=nR$ represents the linear scaling in Shannon transmission, whereas $L(R,n) = 2^{nR}$ represents the exponential scaling in identification via channels ($S=1$)~\cite{ID}.

The above formulations are independent of the channel model and apply to both the noiseless and noisy settings. In Section~\ref{sec:noiseless}, we specialize the BFC problem to a noiseless binary channel, propose an RS tagging construction, and characterize the rate functions for different Boolean function families. In Section~\ref{sec:noisy}, we extend it to noisy channels by concatenating several tagging codewords with a channel code.
\begin{figure}[b]
    \centering
    \includegraphics[width=\columnwidth]{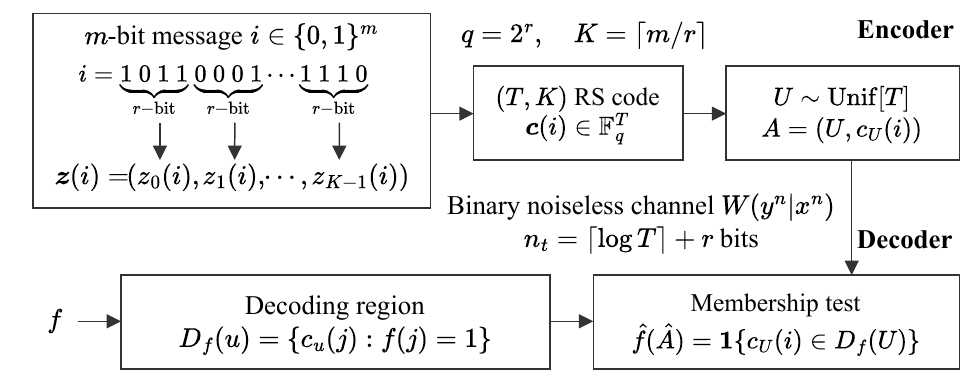}
    \caption{RS tagging construction for BFC over a noiseless channel}
    \label{fig:BFC_code}
\end{figure}
\section{BFC Codes over a Noiseless Channel}
\label{sec:noiseless}
\subsection{RS Tagging Code Construction}
\label{sec:tagging_code}
We adapt the polynomial-hash RS tagging construction in~\cite{universal-hashing,RS-tagging}.
Let the size of the finite field be $q=2^r$ and set $K=\lceil m/r\rceil$. The binary message $i$ of length $m$ is partitioned into $K$ consecutive $r$-bit blocks. Each of these blocks can be interpreted as a symbol over $\F_q$. Therefore, the message can be denoted as $\bm z(i)=(z_0(i),\cdots,z_{K-1}(i))\in \F_q^K$, as shown in Fig.~\ref{fig:BFC_code}. In a practical implementation, zero-padding is applied when $r$ does not divide $m$. We then utilize a $(T,K)$ RS code over $\F_q$ to construct a BFC code in a noiseless channel. Such a code outputs a length-$T$ codeword over $\F_q$ from a length-$K$ sequence over $\F_q$, requiring $1\leq K\leq T\leq q$. Specifically, each message $\bm z(i)$ is associated with a polynomial
\begin{equation}
    P_i(x):=
    \sum_{k=0}^{K-1}z_k(i)x^k.
\end{equation}
For all $N\in\N$, let $[N]:=\{1,\cdots,N\}$. The corresponding RS codeword is given by $\bm c(i)=(c_1(i),\cdots,c_T(i))$, where $c_\ell(i)=P_i(\alpha_\ell)$ and $(\alpha_\ell)_{\ell\in[T]}$ are distinct elements of $\F_q$.

The stochastic BFC encoder draws an index $U\sim\mathrm{Unif}[T]$ and selects RS symbol $c_U(i)$ as tag to form a BFC codeword:
\begin{equation}
    A(i) = (U, c_U(i))\in \mathcal{A}:= [T]\times \F_q.
\end{equation}
Each RS tagging codeword $A(i)$ has a unique binary representation of length $n_t = \lceil\log T\rceil + r$: $\lceil\log T\rceil$ bits for $U$ and
$r$ bits for $c_U(i)$. All logarithms are to base two.
The $n_t$-use noiseless binary bit pipe is characterized by $W(y^{n_t}|x^{n_t}) = \1\{y^{n_t}=x^{n_t}\}$, and has capacity one bit per channel use.

For a chosen Boolean function $f$ at the receiver, we define the symbol-level decoding region at position $u$ as
\begin{equation}
    \label{equ:symbol_decoding_region}
    D_f(u):=\{c_u(j):j\in f^{-1}(1)\}.
\end{equation}
For a received codeword $a=(u,c_u)$, the decoder declares
\begin{equation}
    \hat f(a):=\1\{c_u \in D_f(u)\}.
\end{equation}

\subsection{Finite-Blocklength Results}
After formulating the code construction, we turn to its error analysis. Due to the tagging construction, the error analysis for BFC reduces to a collision analysis between RS codeword symbols. Applying the algebraic structure of RS codes~\cite{RS-code} gives the following finite-blocklength error guarantee.
\begin{proposition}
\label{prop:finite_error}
For $0\leq S\leq 2^m$ and $n_t\geq 2$, consider the RS tagging construction in Subsection~\ref{sec:tagging_code} based on a $(T,K)$ RS code over $\F_q$, satisfying $1\leq K\leq T\leq q$ and $n_t=\lceil\log T\rceil+r$. Then, the construction is an
\begin{equation}
    \left(
        n_t,m,\mathcal F_m(\leq S),
        0,
        \min\left\{1,\frac{S(K-1)}{T}\right\}
    \right)
\end{equation}
BFC code over the noiseless binary channel.
\end{proposition}
\begin{proof}
    Suppose $f(i)=1$, then the tag $c_u(i)\in D_f(u)$ for all $u\in[T]$ by definition of \eqref{equ:symbol_decoding_region}. Therefore, the decoder always declares $\hat{f}(a)=1$ given $a=(u,c_u(i))$, and FN probability is zero.

    Suppose $f(i)=0$, the FP occurs if there exists a message $j\neq i$ and $j\in f^{-1}(1)$ such that
    \begin{equation}
        \label{equ:collision}
        P_i(\alpha_u) = P_j(\alpha_u).
    \end{equation}
    Thus, the FP probability for a pair $(i,f)$ is given by
    \begin{equation}
        \label{equ:fp}
        \begin{split}
            P_{\rm FP}(i,f) & = \Pr\{\hat f(a(i))=1|i,f:f(i)=0\}  \\
            & = \frac{1}{T}\sum_{u=1}^T\1\{\exists\,j\in f^{-1}(1): P_i(\alpha_u) = P_j(\alpha_u)\} \\
            & = \frac{1}{T}\sum_{u=1}^T\1\left\{\bigcup_{j\in f^{-1}(1)}\{P_i(\alpha_u) = P_j(\alpha_u)\}\right\} \\
            & \leq  \frac{1}{T}\sum_{u=1}^T\sum_{j\in f^{-1}(1)}\1\{P_i(\alpha_u) = P_j(\alpha_u)\}
        \end{split}
    \end{equation}
    The property of $(T,K)$ RS code gives: for all $j\neq i$, the difference $P_i(x)-P_j(x)$ is a non-zero polynomial of degree at most $K-1$. Therefore, $\sum_{u=1}^T\1\{P_i(\alpha_u) = P_j(\alpha_u)\}\leq K-1$, yielding
    \begin{equation}
        \begin{split}
            P_{\rm FP}(i,f) & \leq \sum_{j\in f^{-1}(1)}\frac{K-1}{T} \\
            & \leq |f^{-1}(1)|\frac{K-1}{T}\\
            &= \frac{S_f(K-1)}{T} \\
            & \leq \frac{S(K-1)}{T}.
        \end{split}
    \end{equation}
    Combining with $P_{\rm FP}\leq 1$ and taking the maximum over all $(i,f)$ yields
    \begin{equation}
        \label{equ:fp_bound}
        P_{\rm FP} = \max_{i,f:f(i)=0}P_{\rm FP}(i,f)\leq\min\left\{1,\frac{S(K-1)}{T}\right\}.
    \end{equation}
\end{proof}
Proposition~\ref{prop:finite_error} gives an error guarantee for the RS tagging construction in terms of RS code parameters. The bound illustrates that a larger Hamming weight $S$ leads to a higher collision probability, and a larger $K$ carries more information bits but increases collision probability. Increasing $T$ instead adds positional randomization and lowers collision probability. 

Although Proposition~\ref{prop:finite_error} is expressed in terms of RS code parameters, they are implicitly coupled with the blocklength $n_t$. Specifically, $n_t=\lceil\log T\rceil + r$ and $K=\lceil m/r\rceil$, while RS feasibility requires $K\leq T\leq 2^r$.
Therefore, the error probability bound can be relaxed to $S(m)(K-1)/T\leq mS(m)/(rT)$. For fixed $n_t$, minimizing this bound is equivalent to maximizing the denominator $rT=(n_t-\lceil \log T\rceil)T$, subject to $\lceil \log T\rceil\leq n_t-\lceil \log T\rceil$. Since this expression is non-decreasing in $\lceil \log T\rceil$, the optimum is attained at
\begin{equation}
    r=\lceil n_t/2\rceil,\quad T=2^{\lfloor n_t/2\rfloor}.
\end{equation}
Substituting these RS parameters in Proposition~\ref{prop:finite_error} yields the following explicit finite-blocklength tradeoff.

\begin{theorem}
    \label{thm:finite_E2}
    Let $S:\N\to\N$ be non-decreasing and define
    \begin{equation}
        \label{equ:rate_function}
        L_S(R,n):=\max\left\{m\in\N:mS(m)\leq\left\lceil\frac{n}{2}\right\rceil2^{\left\lfloor \frac{n}{2}\right\rfloor}2^{n\left(R-\frac{1}{2}\right)}\right\}.
    \end{equation}
    For all $E_2\in(0,1/2)$ and $n_t\in\N$ with $n_t\geq 2$ such that the set in \eqref{equ:rate_function} is non-empty, set $r=\lceil n_t/2\rceil$, $T=2^{\lfloor n_t/2\rfloor}$, and $m=L_S(1/2-E_2,n_t)$. This yields an
    \begin{equation}
        (n_t, m, \mathcal F_m(\leq S(m)), 0, 2^{-n_tE_2})
    \end{equation}
    BFC code and the finite-blocklength computation rate of this code is $1/2-E_2$ with the rate function $L_S$.
\end{theorem}
\begin{proof}
    By the definition of \eqref{equ:rate_function} and $E_2\in(0,1/2)$, the selected $m$ satisfying
    \begin{equation}
        mS(m)\leq rT 2^{-n_tE_2} < rT.
    \end{equation}
    Since $S(m)\geq 1$, this implies that $m<rT$ and thus $K=\lceil m/r\rceil\leq T$
    With $q=2^r$, the selected $r=\lceil n_t/2\rceil$ and $T=2^{\lfloor n_t/2\rfloor}$ give $T\leq q$. Therefore, these selected RS parameters satisfying $1\leq K\leq T\leq q$ and resulting an RS tagging code of length
    \begin{equation}
        \lceil \log T\rceil + r = \lfloor n_t/2\rfloor + \lceil n_t/2\rceil = n_t.
    \end{equation}
    
    Since $K-1\leq m/r$, Proposition~\ref{prop:finite_error} gives
    \begin{equation}
        \begin{split}
            P_{\rm FP} & \leq \frac{S(m)(K-1)}{T} \\
            & \leq \frac{mS(m)}{rT} \\
            & \leq \frac{rT}{rT}2^{-n_tE_2} \\
            & = 2^{-n_tE_2}.
        \end{split}
    \end{equation}

    The finite-blocklength computation rate of $1/2-E_2$ follows directly from Definition~\ref{def:finite_rate}.
\end{proof}
Based on the error exponent $E_2$, Theorem~\ref{thm:finite_E2} explicitly characterizes the tradeoff between error probability and the supported message length. In particular, a larger $E_2$ yields a faster error decay rate, while reducing the argument $1/2-E_2$ of rate function, resulting in a shorter supported message. 
To illustrate this finite-blocklength tradeoff, we consider identification, rank-test, and exact-weight tasks detailed in Subsection~\ref{sec:noiseless_numerical}.
For fixed $n_t=60$, Fig.~\ref{fig:rate_compare}(a) plots the rate $\log(m)/n_t$ against $E_2$ and shows that increasing $E_2$ exponentially reduces the FP bound $2^{-60E_2}$ while decreasing the computation rate.

\begin{figure}[t]
  \centering
  \includegraphics[width=\columnwidth]{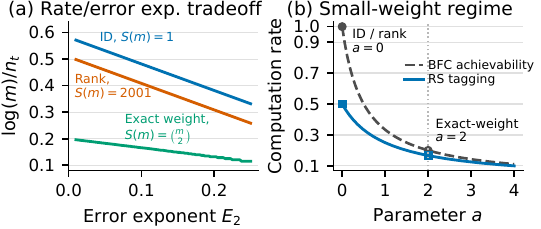}
  \vspace{-0.75cm}
  \caption{Computation rate of RS tagging BFC codes for the noiseless binary channel: (a) Rate $\log(m)/n_t$ vs. error exponent $E_2$ at $n_t=60$, with $m$ chosen using \eqref{equ:rate_function}; (b) Achievable rate $R$ in the small-weight regime with $\log S(m)/\log m\to a$, compared with the BFC achievability~\cite{BFC-capacity}.}
  \label{fig:rate_compare}
\end{figure}

\subsection{Asymptotic Results}

The finite-blocklength tradeoff in Theorem~\ref{thm:finite_E2} immediately suggests an asymptotic construction and gives achievability results. Specifically, for all fixed, arbitrarily small $E_2>0$, the FP probability vanishes exponentially with the blocklength while the computation rate $1/2-E_2$ can be arbitrarily close to $1/2$. In this subsection, we formalize this asymptotic characterization and relate the resulting scaling behavior of the proposed code construction to the fundamental BFC limits in~\cite{BFC,BFC-capacity}. 

\begin{corollary}
    \label{cor:achievable}
    For all $\lambda_1,\lambda_2>0$ and $\eta\in(0,1/2)$, and for all large enough $n_t$, the RS tagging code construction in Subsection~\ref{sec:tagging_code} gives an $(n_t, m, \mathcal F_m(\leq S(m)), \lambda_1,\lambda_2)$ BFC code with $m=L_S(1/2-\eta,n_t)$. Thus, for all small error probabilities, the code supports an asymptotic computation rate arbitrarily close to $1/2$ with the rate function $L_S$ in \eqref{equ:rate_function}. 
\end{corollary}
\begin{proof}
    Fix $\lambda_1,\lambda_2>0$ and $\eta\in(0,1/2)$. Applying Theorem~\ref{thm:finite_E2} with $E_2=\eta$ and $m=L_S(1/2-\eta,n_t)$ yields an
    \begin{equation}
        \left(n_t, m, \mathcal F_m(\leq S(m)), 0, 2^{-n_t\eta}\right)
    \end{equation}
    BFC code whose finite-blocklength computation rate is $1/2-\eta$ with the rate function $L_S$. For sufficiently large $n_t$, $2^{-n_t\eta}\leq\lambda_2$. Since $\lambda_1,\lambda_2$, and $\eta$ are arbitrary, the rate can be made arbitrarily close to $1/2$ while the error probabilities are made arbitrarily small.
\end{proof}

\begin{remark}
In the rest of the paper, we call $R$ an \textit{achievable computation rate}, if the code we construct has a rate of $R$ asymptotically ($n\rightarrow\infty$) with an arbitrarily small error probability. Corollary \ref{cor:achievable} states that $1/2$ is an achievable computation rate with the rate function $L_S$.
\end{remark}

The established rate of $1/2$ arises from the RS parameter choice $r=\lceil n_t/2\rceil$ and $T=2^{\lfloor n_t/2\rfloor}$. Specifically, half of the channel uses transmit RS symbols $c_U(i)$, while the other half transmits the random index $U$. Hence, the channel uses are asymptotically divided equally between the index and the tag.
Importantly, the achievable rate $1/2$ in Corollary~\ref{cor:achievable} is defined with function-dependent rate function $L_S$. To make the result more explicit and comparable to that in~\cite{BFC-capacity}, we express the same scaling behavior using different rate functions associated with the small-, medium-, and large-weight regimes.
\begin{corollary}
    \label{cor:weight_regimes}
    Let $S^{\leftarrow}(x):=\max\{m\in\N:S(m)\leq x\}$. Under the code construction in Subsection~\ref{sec:tagging_code}, as $m\to\infty$
    \begin{itemize}
        \item if $\log S(m)/\log m\to a\geq 0$, then $1/[2(1+a)]$ is an achievable computation rate with the rate function $L_{\rm small}(R,n)=2^{Rn}$.
        \item if $\log S(m)/\log m\to \infty$, then $1/2$ is an achievable computation rate with the rate function $L_{\rm medium}(R,n)=S^{\leftarrow}(2^{Rn})$.
        \item if $\log S(m)/m\to 1$, then $1/2$ is an achievable computation rate with the rate function $L_{\rm large}(R,n)=Rn$.
    \end{itemize}
\end{corollary}
\begin{proof}
    By Corollary~\ref{cor:achievable}, for all $\delta\in(0,1/2)$ and sufficiently large $n$, there exists a BFC code based on the RS tagging construction that supports a message length of at least $L_S(1/2-\delta,n)$ with arbitrarily small error probabilities. 
    By Corollary~\ref{cor:achievable}, to prove a rate $R_0$ is achievable with a rate function $L_\star$, i.e., the code supports computation rates arbitrarily close to $R_0$ with arbitrarily small error probabilities, it is sufficient to show that for all $\eta\in(0,R_0)$ and sufficiently large $n$
    \begin{equation}
        \label{equ:proof_achievability}
        L_\star(R_0-\eta,n)\leq L_S\left(\frac{1}{2}-\delta,n\right).
    \end{equation}

    {\bf Small regime}: When $\log S(m) / \log m \to a$ for $a\geq 0$, consider
    \begin{equation}
        R_0 = \frac{1}{2(1+a)}\quad \text{and}\quad L_{\rm small}(R,n)=2^{Rn}.
    \end{equation}
    For some $\eta\in(0,R_0)$, let
    \begin{equation}
        \tilde m = \left\lceil 2^{n(R_0-\eta)}\right\rceil \geq L_{\rm small}(R_0-\eta,n).
    \end{equation} 
    Since
    \begin{equation}
        \frac{\log\tilde m}{n} \to R_0-\eta\quad \text{and} \quad \frac{\log S(m)}{\log \tilde m} \to a,
    \end{equation}
    we have
    \begin{equation}
        \begin{split}
            \lim_{n\to\infty}\frac{1}{n}\log(\tilde mS(\tilde m)) & = \lim_{n\to\infty}\frac{\log \tilde m}{n}\left(1 + \frac{\log S(\tilde m)}{\log \tilde m}\right) \\
            & = (R_0-\eta)(1+a) \\
            &  = \frac{1}{2}-(1+a)\eta.
        \end{split}
    \end{equation}
    Choose any $0<\delta<(1+a)\eta$, for sufficiently large $n$ we have
    \begin{equation}
        \tilde mS(\tilde m)\leq 2^{n(1/2-\delta)}\leq \frac{n}{2}2^{n(1/2-\delta)}
    \end{equation}
    verifying \eqref{equ:proof_achievability}. Thus, rate $\frac{1}{2(1+a)}$ is achievable with the rate function $L_{\rm small}(R,n)=2^{Rn}$.

    {\bf Medium regime}: When $\log S(m) / \log m\to \infty$, consider
    \begin{equation}
        R_0 = \frac{1}{2}\quad \text{and}\quad L_{\rm medium}(R,n)=S^{\leftarrow}(2^{Rn}).
    \end{equation}
    For some $\eta\in(0,R_0)$, let
    \begin{equation}
        \tilde m = S^{\leftarrow}(2^{n(R_0-\eta)}) \geq L_{\rm medium}(R_0-\eta,n).
    \end{equation} 
    Since
    \begin{equation}
        \frac{\log\tilde m}{\log S(\tilde m)} \to 0\quad\text{and}\quad \frac{\log S(\tilde m)}{n} \leq R_0-\eta,
    \end{equation}
    we have,
    \begin{equation}
        \begin{split}
            \lim_{n\to\infty}\frac{1}{n}\log(\tilde mS(\tilde m)) & = \lim_{n\to\infty}\frac{\log S(\tilde m)}{n}\left(1 + \frac{\log\tilde m}{\log S(\tilde m)}\right) \\
            & \leq (R_0-\eta)(1+0) \\
            & = \frac{1}{2}-\eta.
        \end{split}
    \end{equation}
    Choose any $0<\delta<\eta$, for sufficiently large $n$ we have
    \begin{equation}
        \tilde mS(\tilde m)\leq 2^{n(1/2-\delta)}\leq \frac{n}{2}2^{n(1/2-\delta)}
    \end{equation}
    verifying \eqref{equ:proof_achievability}. Thus, rate $\frac{1}{2}$ is achievable with the rate function $L_{\rm medium}(R,n)=S^{\leftarrow}(2^{Rn})$.

    {\bf Large regime}: When $\log S(m) / m \to 1$, consider
    \begin{equation}
        R_0 = \frac{1}{2}\quad \text{and}\quad L_{\rm larage}(R,n)=Rn.
    \end{equation}
    For some $\eta\in(0,R_0)$, let
    \begin{equation}
        \tilde m = \lceil n(R_0-\eta)\rceil \geq L_{\rm large}(R_0-\eta,n).
    \end{equation} 
    Since
    \begin{equation}
        \frac{\log\tilde m}{n} \to 0\quad, \frac{\log S(m)}{m} \to 1,\quad  \text{and}\quad \frac{\tilde m}{n}\to R_0-\eta,
    \end{equation}
    we have,
    \begin{equation}
        \begin{split}
            \lim_{n\to\infty}\frac{1}{n}\log(\tilde mS(\tilde m)) & = \lim_{n\to\infty}\left(\frac{\log \tilde m}{n} + \frac{\log S(\tilde m)}{\tilde m}\frac{\tilde m}{n}\right) \\
            & = 0 + 1\cdot (R-\eta) \\
            & = \frac{1}{2}-\eta.
        \end{split}
    \end{equation}
    Choose any $0<\delta<\eta$, for sufficiently large $n$ we have
    \begin{equation}
        \tilde mS(\tilde m)\leq 2^{n(1/2-\delta)}\leq \frac{n}{2}2^{n(1/2-\delta)}
    \end{equation}
    verifying \eqref{equ:proof_achievability}. Thus, rate $\frac{1}{2}$ is achievable with the rate function $L_{\rm large}(R,n)=Rn$.
\end{proof}
Corollary~\ref{cor:weight_regimes} explicitly shows that the rate function for BFC problem strongly depends on the Boolean function Hamming weight. A polynomial growth $S(m)=m^a$ allows exponential scaling of the supported message length. When $S(m)$ grows faster yet sub-exponential, its inverse sets the scaling. When $S(m)$ grows exponentially as $2^m$, the supported message length scales linearly with $n$, as in conventional transmission.

For a noiseless binary channel,~\cite{BFC-capacity} establishes achievable rates of $1/(1+2a)$ in the small-weight regime, $1/2$ in the medium-weight regime, and $1$ in the large-weight regime, with corresponding rate functions. 
In comparison, Corollary~\ref{cor:weight_regimes} gives rates of $1/[2(1+a)]$, $1/2$, and
$1/2$, respectively, under our RS tagging construction.
Fig.~\ref{fig:rate_compare}(b) compares the two small-weight achievable rates as functions of parameter $a$. For identification and rank-test ($a=0$) tasks, our code construction achieves one half of the theoretical capacity limit. For the exact-weight task ($a=2$), our construction attains a rate of $1/6$, whereas the achievable rate of $1/5$ is reported in~\cite{BFC-capacity}.
In the medium-weight regime, where $a\to\infty$, our explicit code attains the achievable rate of $1/2$ in~\cite{BFC-capacity} under the rate function $L_{\rm medium}$.
For the large-weight regime, our construction only achieves a rate of $1/2$ with rate function $L_{\rm large}$, whereas transmitting the entire $m$-bit message attains rate $1$ under $L_{\rm large}$. Thus, our construction is uninteresting for this regime, as the conventional message transmission gives a better rate.

\subsection{Numerical Results}
\label{sec:noiseless_numerical}
We consider three Boolean functions: identification with a fixed target, rank-test $f(i)=\1\{{\rm int}(i)\leq 2000\}$ (${\rm int}(\cdot)$ denotes the nonnegative integer represented by the binary string), and the exact-weight $f(i)=\1\{\|i\|_1=2\}$ ($\|\cdot\|_1$ denotes the Hamming weight of the binary string). The message lengths are fixed at $m=100{,}000$, $150$ and $100$, respectively, giving $S=1$, $2001$, and $\binom{100}{2}=4{,}950$.
For $n_t\in\{28,30,\cdots,46\}$, we set $r=n_t/2$, $T=2^r$, and $K = \lceil m/r\rceil$.
We uniformly sample $2{,}000$ negative messages $i\notin f^{-1}(1)$, and enumerate all $T$ indices to compute the FP probability. 

Fig.~\ref{fig:noiseless_noisy}(a) compares the average and maximum FP probabilities with the bound $\min\{1,S(K-1)/T\}$ in Proposition~\ref{prop:finite_error}. For each fixed task, FP probabilities decrease as the channel uses increase. The bound is generally conservative for these Boolean function families. The plotted maximum is over sampled messages, not all possible messages and functions.

\section{Concatenated Extension to Noisy Channels}
\label{sec:noisy}


The RS tagging construction in Subsection~\ref{sec:tagging_code} gives an $n_t=\lceil\log T\rceil + r$ bit codeword $A(i)$. To protect it over a noisy channel, we concatenate it with a channel code~\cite{RS-tagging}. However, this can be highly inefficient since for practical applications, $n_t$ is considerably smaller than the channel code's information block. Thus, we pack tagging codewords from independent messages into each block. Next, we formally introduce the packed concatenated RS tagging code construction and derive its error bounds and asymptotic achievable rate.

\subsection{Packed Concatenated RS Tagging Code Construction}
\label{sec:packing_concatenation}
Consider a binary channel and an $(N_b,N_i,\xi)$ binary channel code with an encoder $\phi_{\rm c}:\{0,1\}^{N_i}\to\{0,1\}^{N_b}$ and a decoder $\psi_{\rm c}:\{0,1\}^{N_b}\to\{0,1\}^{N_i}$, where $N_b$ denotes the blocklength, $N_i$ denotes the information length, and the maximal frame error rate is $\xi$. The channel code rate is $R_c = N_i/N_b$.

Each RS tagging codeword takes $n_t$ bits. For simplicity, we assume $n_t$ divides $N_i$. Hence, one information block can pack $G=N_i/n_t\in\N$ independent tagging codewords. In practical implementations, zero padding will be applied if $n_t$ does not divide $N_i$. Therefore, the effective number of channel uses is
\begin{equation}
    \label{equ:eff_channel_use}
    n_{\rm eff} = \frac{N_b}{G}=\frac{n_t}{R_c}.
\end{equation}
Here, $N_b$ is the physical blocklength of the packed channel code, whereas $n_{\rm eff}=N_b/G$ is a normalization per BFC task.

For each slot $g\in[G]$, let the message be $i_g\in\mathcal{I}_m$. We generate the corresponding noiseless BFC codeword $A_g=(U_g,c_{U_g}(i_g))$ described in Subsection~\ref{sec:tagging_code}. The binary representations of $(A_1,\cdots,A_G)$ are concatenated in a known order, yielding an $N_i$-bit information block $V\in\{0,1\}^{N_i}$. Then, the channel encoder transmits $X^{N_b}=\phi_{\rm c}(V)$.

Given the channel output $Y^{N_b}\sim W(\cdot|X^{N_b})$, the channel decoder computes $\psi_{\rm c}(Y^{N_b})$ and recovers $G$ tagging codewords $\hat A_g=(\hat U_g,\hat C_g)$, $g\in[G]$. For each slot, the receiver can select a function $f_g\in\mathcal F_m$. Utilizing the symbol-level decoding region $D_{f_g}$ defined in \eqref{equ:symbol_decoding_region}, the decision for slot $g$ is 
\begin{equation}
    \hat f_g=\1\{\hat U_g\in[T], \hat C_g\in D_{f_g}(\hat U_g)\}.
\end{equation}
If an index $\hat U_g$ does not correspond to any element of $[T]$ due to channel distortion, the decoder declares zero. Note that if we choose $T = 2^t$ for some integer $t\leq r$, then the binary sequence $\hat U_g$ always corresponds to one element in $[T]$.

\subsection{Finite-Blocklength and Asymptotic Performance}
The packing construction does not give an ordinary single-task BFC code of blocklength $n_{\rm eff}$. Instead, $N_b$ channel uses jointly serve $G$ independent tasks and $n_{\rm eff}$ represents the effective channel uses per task. Therefore, we utilize the notion of \textit{effective BFC code} to state per-task error guarantees.
\begin{proposition}
    \label{prop:noisy_finite_error}
    Consider the packed concatenated RS tagging construction that concatenates $G$ noiseless BFC codewords in Proposition~\ref{prop:finite_error} using an $(N_b,N_i,\xi)$ channel code. For each function evaluation, the construction is an 
    \begin{equation}
        \left(n_{\rm eff},m,\mathcal F_m(\leq S),\xi,\min\left\{1,\frac{S(K-1)}{T}+\xi\right\}\right),
    \end{equation}
    effective BFC code, where $n_{\rm eff}$ is defined in \eqref{equ:eff_channel_use}
\end{proposition}
\begin{proof}
    Let $\mathcal{E}:=\{\hat V\neq V\}$ denote the channel decoding error event. The $(N_b,N_i,\xi)$ channel code definition implies $\Pr\{\mathcal{E}\}\leq\ \xi$. 
    We fix an arbitrary slot $g\in[G]$. Suppose $f_g(i_g)=1$, a correct channel decoding gives $\hat A_g=A_g$, and $\hat f_g=1$ according to Proposition~\ref{prop:finite_error}. 
    Therefore, 
    \begin{equation}
        \{\hat f_g=0\}\subseteq \{\hat A_g\neq A_g\}\subseteq \mathcal{E},
    \end{equation}
    and consequently the FN probability for slot $g$ is
    \begin{equation}
        P_{\rm FN}^{(g)} = \Pr\{\hat f_g=0|f_g(i_g)=1\} \leq \xi.
    \end{equation}

    Suppose $f_g(i_g)=0$ and let $\mathcal{C}_g$ denote the noiseless RS tag collision event in slot $g$. Specifically, 
    \begin{equation}
        \mathcal{C}_g:=\{\exists\,j\in f_g^{-1}(1):P_{i_g}(\alpha_{U_g})=P_j(\alpha_{U_g})|f_g(i_g)=0\}.
    \end{equation}
    Since $\{\hat f_g=1\}\cap \mathcal{E}^{\rm c}\subseteq \mathcal{C}_g$ and $\{\hat f_g=1\}\cap \mathcal{E}\subseteq \mathcal{E}$, the FP probability for slot $g$ is
    \begin{equation}
        \begin{split}
            P_{\rm FP}^{(g)} 
            & = \Pr\{\hat f_g=1\} \\
            & = \Pr\{\{\hat f_g=1\}\cap \mathcal{E}^{\rm c}\} + \Pr\{\{\hat f_g=1\}\cap \mathcal{E}\} \\
            & \leq \Pr\{\mathcal{C}_g\} + \Pr\{\mathcal{E}\}  \\
            & \leq \frac{S(K-1)}{T}+\xi.
        \end{split}
    \end{equation}
    
    Both the FN and FP probability bounds are uniform for all slots, messages and Boolean functions, taking the maximum and imposing an upper bound of one prove the proposition.
\end{proof}

Proposition~\ref{prop:noisy_finite_error} preserves the error structure of the noiseless construction. Given correct channel decoding, the FN and FP probabilities remain as Proposition~\ref{prop:finite_error}. A channel failure adds an error probability at most $\xi$. This additive structure reflects a separate source-channel coding scheme. Moreover, packing does not introduce a factor $G$ in the error analysis, and thus does not affect the worst-case error of an individual task.

Consequently, the finite-blocklength tradeoff developed in Theorem~\ref{thm:finite_E2} can be reused in the noisy channel setting. 
Since $n_t=n_{\rm eff}R_c$, we apply Theorem~\ref{thm:finite_E2} with blocklength $n_{\rm eff}R_c$ and represent the rate function in terms of the effective channel use $n_{\rm eff}$. Incorporating the additional channel decoding error $\xi$ yields the following finite-blocklength result.
\begin{corollary}
    \label{cor:noisy_finite_E2}
    Define $L_{S,R_c}(R,n) := L_S(R/R_c,nR_c)$.
    Suppose $r$ and $T$ are chosen as in Theorem~\ref{thm:finite_E2}, and the channel code satisfies $\xi\leq 2^{-N_b E_{\rm ch}}\leq 1$, for some $E_{\rm ch}>0$.
    For all $E_2\in (0,R_c/2)$, set $m=L_{S,R_c}(R_c/2-E_2,n_{\rm eff})$, then the packed concatenated RS tagging construction gives an $(n_{\rm eff}, m, \mathcal F_m(\leq S(m)), \lambda_1,\lambda_2)$ effective BFC code where
    \begin{equation}
        \lambda_1= 2^{-N_bE_{\rm ch}},\quad \lambda_2=\min\{1, 2^{-n_{\rm eff}E_2}+2^{-N_bE_{\rm ch}}\},
    \end{equation}
    and the finite-blocklength computation rate of this code is $R_c/2-E_2$ with the rate function $L_{S,R_c}$.
\end{corollary}
\begin{proof}
    Let $\tilde E_2=E_2/R_c\in(0,1/2)$. Since $n_t=n_{\rm eff}R_c$, 
    the rate function defined in \eqref{equ:rate_function} is
    \begin{equation}
        L_S\left(\frac{1}{2}-\tilde E_2,n_t\right) = L_{S,R_c}\left(\frac{R_c}{2}-E_2, n_{\rm eff}\right).
    \end{equation}
    Applying Theorem~\ref{thm:finite_E2} to each inner RS tagging construction with blocklength $n_t$ gives zero FN probability and a FP probability bounded by $2^{-n_{\rm eff}E_2}$. 
    Combining with Proposition~\ref{prop:noisy_finite_error} yields the probability bounds.

    Finally, the finite-blocklength rate of $R_c/2-E_2$ directly follows from Definition~\ref{def:finite_rate}.
\end{proof}

Similar to Theorem~\ref{thm:finite_E2}, the error exponent $E_2$ provides a tradeoff between error and rate. The factor $R_c$ implies that only a fraction of the physical channel uses are available to carry RS tagging bits. Asymptotically, the achievable computation rate over the noisy channel is thus scaled by the channel code rate. Formally, we have the following asymptotic result.

\begin{corollary}
    Consider packed concatenated RS tagging constructions with fixed channel code rate $R_c$ and $n_{\rm eff}\to\infty$. If the channel code satisfies $\xi\to0$, then $R_c/2$ is an achievable computation rate with the rate function $L_{S,R_c}$. 
\end{corollary}
\begin{proof}
    Fix $\lambda_1,\lambda_2>0$ and $\eta\in(0,R_c/2)$. Applying Theorem~\ref{thm:finite_E2} with $E_2=\eta/R_c$ to each inner code, followed by Proposition~\ref{prop:noisy_finite_error}, yields an 
    \begin{equation}
        (n_{\rm eff}, m, \mathcal F_m(\leq S(m)), \xi, \min\{1,2^{-n_{\rm eff}\eta}+\xi\})
    \end{equation}
    effective BFC code with finite-blocklength computation rate $R_c/2-\eta$ with the rate function $L_{S,R_c}$. 
    Since $\xi\to0$, for sufficiently large $n_{\rm eff}$,
    \begin{equation}
        \xi \leq \lambda_1,\quad 2^{-n_{\rm eff}\eta}+\xi\leq\lambda_2.
    \end{equation}
    Since $\lambda_1,\lambda_2$, and $\eta$ are arbitrary, $R_c/2$ is an achievable computation rate under the rate function $L_{S,R_c}$..
\end{proof}
Concatenation separates the two sources of rate loss: the factor $R_c$ arises from reliable transmission over channels, while the remaining $1/2$ is inherent to the RS tagging construction.
Crucially, if a capacity-achieving sequence of channel codes is employed so that $R_c$ approaches capacity $C$, the code construction achieves an asymptotic computation rate of $C/2$.

\begin{figure}[t]
  \centering
  \includegraphics[width=\columnwidth]{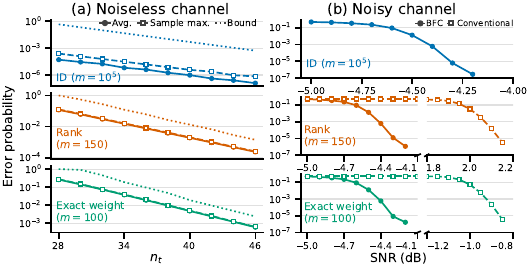}
  \vspace{-0.75cm}
  \caption{BFC task performance. (a) Noiseless FP probabilities; (b) Error probabilities over AWGN for packed BFC vs. conventional message transmission.}
  \label{fig:noiseless_noisy}
\end{figure}
\subsection{Numerical Results}
\label{sec:noisy_numerical}
We utilize the same three tasks and message lengths as in Subsection~\ref{sec:noiseless_numerical} with $n_t=60$, $G=360$, and $N_b=64{,}800$. Therefore, RS parameters are $r=30$, $T=2^{30}$, and $K=3334$, $5$, and $4$, respectively. The packed tagging codewords of length $N_i=21{,}600$ are encoded by a DVB-S2~\cite{DVB-S2} LDPC code of rate $1/3$. We transmit BPSK symbols over an AWGN channel and apply belief-propagation decoding with at most $50$ iterations. The effective channel uses per BFC task are $n_{\rm eff}=180$.
We simulate $10^5$ frames. In each frame, we uniformly sample $180$ positive and $180$ negative messages. Fig.~\ref{fig:noiseless_noisy} shows the average error probabilities over all messages and frames.

For comparison, conventional communication transmits the complete $m$-bit message, applies a channel code, and then evaluates the same Boolean function on the decoded message.
Both schemes use $N_b=64{,}800$ channel uses for $G=360$ tasks (same effective channel use) and the same BPSK symbol energy. For rank-test, the $150$-bit messages fill a rate-$5/6$ LDPC information block of $54{,}000$ bits. For exact-weight, the $100$-bit messages provide $36{,}000$ information bits in a rate-$3/5$ LDPC information block of $38{,}880$ bits and the remaining bits are padding. 
For identification task, the full-messages transmission is infeasible as it would require $100{,}000$ channel uses even if the channel is noiseless. For conventional scheme, we simulate $10^3$ to $3\times10^{7}$ frames, depending on SNR.

Fig.~\ref{fig:noiseless_noisy}(b) shows that the task error probability drops sharply between $-4.6$ and $-4.2$~dB. This waterfall is mainly governed by the FN events that are caused by channel decoding failures, whereas the error floor in the high SNR regime is determined by RS tag collisions.
Crucially, our construction achieves reliable task computation at a substantially lower SNR. At error probability $10^{-5}$, it yields a $6.4$~dB gain for rank-test task and a $3.4$~dB gain for exact-weight task.
These gains stem from conveying task-relevant information through a short tagging code. Using a $60$-bit codeword rather than the $150$-bit and $100$-bit messages for rank-test and exact-weight tasks enables a much lower rate and more reliable channel code.

{\bf Limitation:} Although our coding scheme achieves the optimal scaling behavior (i.e., rate function), the suboptimal rate is due to the single-layer tagging construction. For identification ($S=1$), the rate can be improved by adopting a two-layer code construction, yielding identification capacity-achieving codes~\cite{CWC,RS-performance}. However, such a construction does not directly extend to BFC when $S$ grows with $m$. Explicit capacity-achieving BFC constructions are left for future work.


\begin{thebibliography}{99}

\bibitem{beyond-bits}
D. G\"und\"uz, Z. Qin, I. E. Aguerri, H. S. Dhillon, Z. Yang, A. Yener, K.-K. Wong, and C.-B. Chae, ``Beyond transmitting bits: Context, semantics, and task-oriented communications,'' \emph{IEEE J. Sel. Areas Commun.}, vol. 41, no. 1, pp. 5--41, Jan. 2023.

\bibitem{task-oriented}
J. Shao, Y. Mao, and J. Zhang, ``Learning task-oriented communication for edge inference: An information bottleneck approach,'' \emph{IEEE J. Sel. Areas Commun.}, vol. 40, no. 1, pp. 197--211, Jan. 2022.

\bibitem{propositional-logic}
J. Bao, P. Basu, M. Dean, C. Partridge, A. Swami, W. Leland, and J. A. Hendler, ``Towards a theory of semantic communication,'' in \emph{Proc. IEEE Netw. Sci. Workshop (NSW)}, 2011, pp. 110--117.

\bibitem{BFC}
J. Zhu and M. Frey, ``Beyond identification: Computing Boolean functions via channels,'' in \emph{Proc. IEEE Int. Symp. Inf. Theory (ISIT)}, 2026, pp. 1--6.

\bibitem{ID}
R. Ahlswede and G. Dueck, ``Identification via channels,'' \emph{IEEE Trans. Inf. Theory}, vol. 35, no. 1, pp. 15--29, Jan. 1989.

\bibitem{ID-task}
Y. Zhao, H. Boche, and C. Deppe, ``Joint identification and sensing with noisy feedback: A task-oriented communication framework for 6G,'' arXiv:2603.29649, Mar. 2026.

\bibitem{BFC-capacity}
J. Zhu and M. Frey, ``Capacity regimes for Boolean function computation via channels,'' arXiv:2608.10816, Aug. 2026.

\bibitem{CWC}
S. Verd\'u and V. K. Wei, ``Explicit construction of optimal constant-weight codes for identification via channels,'' \emph{IEEE Trans. Inf. Theory}, vol. 39, no. 1, pp. 30--36, Jan. 1993.

\bibitem{universal-hashing}
K. Kurosawa and T. Yoshida, ``Strongly universal hashing and identification codes via channels,'' \emph{IEEE Trans. Inf. Theory}, vol. 45, no. 6, pp. 2091--2095, Sep. 1999.

\bibitem{RS-tagging}
P. Moulin and R. Koetter, ``A framework for the design of good watermark identification codes,'' in \emph{Proc. SPIE, Security, Steganography, Watermarking Multimedia Contents VIII}, vol. 6072, 2006.

\bibitem{RS-performance}
S. Derebeyo\u{g}lu, C. Deppe, and R. Ferrara, ``Performance analysis of identification codes,'' \emph{Entropy}, vol. 22, no. 10, Art.\ no. 1067, Oct. 2020.

\bibitem{RS-implementation}
R. Ferrara, L. Torres-Figueroa, H. Boche, C. Deppe, W. Labidi, U. J. M\"onich, and V. C. Andrei, ``Implementation and experimental evaluation of Reed--Solomon identification,'' in \emph{Proc. 27th Eur. Wireless Conf. (EW)}, 2022, pp. 7--12.

\bibitem{ID-tutorial}
C. von Lengerke, J. A. Cabrera, M. Reisslein, and F. H. P. Fitzek, ``Codes for identification via channels: Tutorial for communications generalists,'' \emph{IEEE Commun. Surveys Tuts.}, vol. 28, no. 1, pp. 181--223, Jan. 2026.


\bibitem{RS-code}
R. M. Roth, \emph{Introduction to Coding Theory}. Cambridge, U.K.: Cambridge Univ. Press, 2006.

\bibitem{DVB-S2}
European Telecommunications Standards Institute,
``Digital Video Broadcasting (DVB); Second generation framing
structure, channel coding and modulation systems for broadcasting,
interactive services, news gathering and other broadband satellite
applications; Part 1: DVB-S2,''
ETSI EN 302 307-1 V1.4.1, Nov. 2014.

\end{thebibliography}
\end{document}